\documentclass{article}
 
\usepackage{amsmath}
\usepackage{amsthm}
\usepackage{cite}

\newtheorem{theorem}{Theorem}
\newtheorem{lemma}{Lemma}
\title{Approximation algorithms for Capacitated Facility Location Problem with Penalties}
\author{Neelima Gupta \footnote{ Department of Computer Science, University of Delhi, India, ngupta@cs.du.ac.in}
       \and Shubham Gupta, \footnote{Department of Computer Science and Engineering, IIT Kharagpur, India, shubham.gupta@cse.iitkgp.ernet.in}}

\begin{document}
\maketitle

\begin{abstract}
In this paper, we address the problem of capacitated facility location problem
with penalties (CapFLPP) paid per unit of unserved demand. In case of uncapacitated FLP with penalties
demands of a client are either entirely met or are entirely rejected and penalty is paid. In the uncapacitated case, there is no reason to serve a client partially. Whereas, in case of CapFLPP, it may be beneficial to serve a client partially instead of not serving at all and,  pay the penalty for the unmet demand. Charikar et. al.~\cite{charikar2001algorithms}, Jain et. al.~\cite{jain2003greedy} and Xu- Xu~\cite{xu2009improved} gave $3$, $2$ and 
$1.8526$ approximation, respectively, for the uncapacitated case . We present $(5.83 + \epsilon)$ factor for the case of uniform capacities and $(8.532 + \epsilon)$ factor for non-uniform capacities. 
\end{abstract}

\section{Introduction}

 We consider the  capacitated facility location problem with penalties(CapFLPP). We are given a set $C$ of clients  and a set $F$ of facilities .  Each facility $i$ has a facility opening cost $f_i$ and capacity $u_i$, which is the maximum amount of demand it can serve. Each client $j$ has some demand $d_j$ to be served by open facilities and penalty $p_j$ paid per unit for the unserved demand. $c_{ij}$ is the cost of serving one unit of  demand of $j$ by facility $i$. A facility can serve a demand only if it is opened. Service costs $c_{ij}$ are metric. Demands are splittable i.e. the demands of a client may be served by different facilities. Objective is to select a subset of facilities to open, select and assign the selected demands to open facilities such that the facility capacities are not violated and the total cost of opening the facilities plus the assignment cost plus the penalties paid by the unserved demands is minimized. Here the penalties are paid per unit of unserved demand. This is in contrast to the notion of penalties used in the previous work (see~\cite{charikar2001algorithms,jain2003greedy,xu2009improved} for facility location with penalties and $k$ - median with penalties, ~\cite{goemans1995general, bienstock1993note}  for prize-collecting TSP and steiner tree problem and ~\cite{hajiaghayi2010prize} for Prize Collecting Steiner Network Problems.)
 In~\cite{charikar2001algorithms} Charikar et al. gave a $3$ factor algorithm for the uncapacitated FLPP (UFLPP) wherein a client is either served completely or it is not served at all. This was improved to factor $2$ in~\cite{jain2003greedy}. The best known factor of $1.8526$ was given by Xu and Xu in~\cite{xu2009improved}. 
 In case there are no capacities on the facilities, there is no reason why a client should be served partially. But in case of capacities, it may be beneficial
  to serve a client partially, to the extent the capacities permit, instead of not serving at all, in case the penalty of not servicing is high. Consider a car manufacturing company with a target to churn out $N$ cars per day. A shortfall of $1$ car costs the company Rs. $p$. Also, consider a warehouse that supplies a particular part to the company. If the warehouse is not able to meet the demand of the company on some day, company charges the owner of the warehouse a penalty per unit of the cars that are delayed due to this. Here, facilities are a collection of warehouses and clients are the manufacturing companies.
 
 The work on UFLPP extends the lp- based techniques like LP-rounding, primal-dual and dual-fitting, used to approximate UFLP (without penalties)  to incorporate penalties.  For example Xu and Xu use the primal-dual technique of~\cite{charikar2001algorithms} followed by greedy local search to handle penalties. As the integrality gap of the natural LP-relaxation for the capacitated FLP(CapFLP)  is arbitrarily large, we extend the local search technique of Chudak-Williamson (henceforth referred to as CW)~\cite{chudak} for uniform capacities and Pal-Tardos-Wexler (henceforth referred to as PTW)~\cite{paltree} in case the capacities are not uniform. We present $(6+\epsilon)$ factor for Uniform Facility Location Problem with Penalties (UnifFLPP) which leads to $(5.83 + \epsilon)$ after scaling and $(9+\epsilon)$ factor for Non-uniform Facility Location Problem with Penalties (CapFLPP) which leads to $(8.532 + \epsilon)$ after scaling.
 
 One may think of adding another dummy facility say $N$ with zero facility opening cost, large capacity ($\sum_j d_j$) and $p_j$ as the service cost of serving client $j$ from $N$ and then applying CW/PTW. But this cannot be done as the penalties are not metric. We present the following two results:
 
  \begin{theorem} 
  \label{theo-unif}
  There is a polynomial time local search algorithm that provides a locally optimal solution of cost at most $(5.83 + \epsilon)$ times the optimum cost for uniform capacity facility location problem with penalties (UnifFLPP).
 \end{theorem}

 \begin{theorem} 
\label{theo-nonunif}
 There is a polynomial time local search algorithm that provides a locally optimal solution of cost at most $(8. 532 + \epsilon)$ times the optimum cost for nonuniform capacity facility location problem with penalties (CapFLPP).
 \end{theorem}
  
 The remainder of the paper is organized as follows: In section 2 we present our result for the uniform capacities. We first briefly review the non-penalty UnifFLP of Chudak-Williamson which we extend to incorporate penalties. Section 3 presents the results for the case of non-uniform capacities. Again we briefly review the non-penalty CapFLP of  Pal-Tardo-Wexler  which is extended to handle penalties. Finally, section 4 concludes the paper by throwing some light on the future work to improve the factor by applying our technique to other (better) algorithms for UnifFLP and CapFLP.  
 
 \section{Uniform Capacities}
 In this section, we consider the facility location problem with penalties when the capacities of the facilities are uniform {\em i.e.} $u_i = U ~\forall ~i \in F$. We first briefly review the local search algorithm of Chudak-Williamson\cite{chudak} and then show how it can be extended to integrate penalties. Operations for UnifFLPP largely remain the same as that in\cite{chudak}  except that we use the min-cost flow with penalties, instead of min-cost flow, to assign the demands to a given set of facilities. 
 
 \label{unif}
\subsection{Preliminaries and Previous Work}
For the capacitated facility location problem, given a subset $S\subseteq F$, the optimal assignment of clients to the set of facilities in $S$ can be computed by solving a mincost flow (MCF) problem. Therefore to solve the problem, we only need to determine a good subset $S\subseteq F$ of facilities to open.  With the abuse of notation,  $S$ denotes both the solution as well as the set of facilities opened in the solution. Cost of solution $S$ is denoted by $c(S) = c_f(S) + c_s(S)$, where $c_f(S)$ is the facility cost and $c_s(S)$ is the service cost of the solution $S$.

Chudak-Williamson~\cite{chudak} suggested a local search algorithm to find a good approximate solution for UnifFLP (without penalties). Starting with a feasible solution $S$ the following operations are performed to improve the solution if possible.

\begin{itemize}
\item  {\textbf{add(s):}}  $S\gets S\cup \{s\}, s\not\in S$. 
\item  {\textbf{delete(s):}} $S\gets S \setminus \{s\} , s \in S$.  
\item  {\textbf{swap(s, t):}} $S\gets S\cup  \{t\} \setminus \{s\}, s \in S, t \notin S$. 

\end{itemize} 

Add operation is used to bound the service cost as per the following Lemma.

\begin{lemma}[\cite{Korupolu,chudak}]
\label{lem-serv}
The service cost $c_s(S)$  of a locally optimal solution $S$ is at most the total cost $c(S^*)$ of the optimum solution $S^*$ i.e.
\begin{center}
\vspace{-0.4cm}
$c_s(S) \leq c_s(S^*) + c_f(S^*)  =  c(S^*)$.
\end{center}
\end{lemma}

\textbf{Path Decomposition}: Standard technique of path decomposition provides a feasible way of re-assigning demands from $s \in S$ to a facility $t \in F$. Let $(S, x)$ be a solution where  $x(s,j)  $ denote the demand of client $j$ served by the facility $s$ in $S$.  The total demand of client $j$ served in $ S$ is then $\sum_{s \in S} x(s,j) $, denoted by $x(\cdot, j)$. We can think of the assignment $x$ as a flow in a bipartite graph with vertices corresponding to facilities and clients, where $x(s,j) $ units flow from a facility $s$ to a client $j$.  To compare the locally optimal solution $(S, x)$ with the optimal solution $(S^*, x^*)$  a flow of $x-x^*$ is considered, i.e. the flow on the edge $(s,j)$ is $x(s,j)-x^*(s,j)$. A negative flow indicates flow in a reverse direction i.e. from a client to a facility. By standard path stripping arguments,  the flow is decomposed into a set of paths $P$ and cycles $C$. Each path starts at a vertex in $S$ and ends at a vertex in $S^*$. 
Let $w(p)$ denote the amount of flow on path $p$ and  $c(p)$ denote its cost, the sum of edge costs for each edge in the path.  For any subset of paths $P' \subseteq P$: $w(P')=\sum_{ p \in P'}w(p)$ and  its value $v(P')=\sum_{p \in P'} w(p)c(p)$ .
Then clearly,  $\sum_P c(P)w(P) \leq c_s(S) + c_s(S^*)$.

Using path decomposition, the notion of {\em transfer paths}, {\em swap paths}  and that of {\em heavy} and {\em light} facilities are defined. \textit{Transfer paths} are paths from $S \setminus S^*$ to $S \cap S^*$. They are used to transfer the demands served by a facility $s \in S \setminus S^*$ to other facilities in $S$, when either $s$ itself is closed  or some other facility $i \in S$ is closed and its demands are transferred to $s$. \textit{Swap paths} are paths from $S \setminus S^*$ to $S^* \setminus S$. They are used to transfer the demands of $s \in S \setminus S^*$ to a facility $t \in S^* \setminus S$.  
A facility in   $S\setminus S^*$ is said to be {\em heavy} if the total weight of the swap paths emanating from it is more than $U/2$ and it is called {\em light} otherwise. Let $H$ and $L$ be the set of {\em heavy} and {\em light} facilities respectively in $S \setminus S^*$. Let $v(Sw(H, \cdot))$  denote the cost of the swap paths and $v(Tr(H, \cdot))$ denote the cost of the transfer paths for facilities in $H$. Similarly, let $v(Sw(L, \cdot))$  denote the cost of the swap paths and $v(Tr(L, \cdot))$ denote the cost of the transfer paths for facilities in $L$.

 Two transshipment problems are set up, one between the set $H$ of  heavy facilities and $S^* \setminus S$  and another between the set $L$ of light facilities and $S^* \setminus S$. 
The increase in service cost by reassigning a demand of client $j$ from a facility $s \in S \setminus S^*$, (to a facility $t \in S^* \setminus S$, is $c_{tj}-c_{sj}$ which is bounded from above by $c_{st}$ using triangle inequality.  
When a facility $s$ is swapped with another facility $t$, at most $U$ demands  of $s$ need to be reassigned to $t$ leading to a reassignment cost of at most $U c_{st}$. The cost of closing $s$ and opening $t$ is $f_t - f_s$. The overall cost of the swap operation is then bounded by  $U c_{st} + f_t - f_s$. Let $\hat{c}_{st}$ denote this cost.

The transshipment problem is set up so that each facility in $H \setminus S^*$ is closed exactly once and each facility in $S^* \setminus S$ is opened at most twice.
Let $y(s, t)$ denote the fraction of  demands served by $s$  assigned to $t$. Thus the problem is formulated as
 \[ minimize \sum\limits_{s \in H, t \in S^*\setminus S} \hat{c}_{st} y(s,t) \]
 \vspace{-0.3cm}
 subject to
 \begin{align*}
  \vspace{-0.4cm}
\sum\limits_{t \in S^* \setminus S } y(s, t) = 1 , \quad  \forall s \in H  \quad (1)\\
\sum\limits_{s \in H } y(s,t) \leq  2 , \quad\forall t\in S^*\setminus S  \quad (2)\\
y(s,t) \geq 0 ,\quad \forall s \in H, t\in S^*\setminus S,  \quad (3)
\end{align*}

\begin{lemma}[~\cite{chudak}]

There is a fractional solution to the transshipment problem of cost no more than $2 v(Sw(H, .)) + 2 c_f (S^* \setminus S) - c_f(H) $.
\end{lemma}
\begin{proof}
Swap paths provide a flow whose cost is no more than $2 v(Sw(H, \cdot)) + 2 c_f (S^* \setminus S) - c_f(H)$. 
\end{proof}



 


From the integrality of the transshipment polyhedra, there exists an integral solution to the transshipment problem of cost no more than this. When $y(s, t) = 1, $, we perform $swap(s, t)$ assigning $w(Sw(s, \cdot))$ demands of $s$ to $t$ at cost at most $\hat{c}_{st}$.  Assign the remaining $w(Tr(s, \cdot))$ demands of $s$ to other facilities in $S$ at a cost of $v(Tr(s, \cdot))$.  Thus we get the following lemma:

\begin{lemma}[\cite{chudak}]
If no  swap operation is admissible, then
\label{lem:unif-heavy}
\begin{center}
\vspace{-.3cm}
$c_f(H) \le 2 c_f (S^* \setminus S) + 2 v(Sw(H, \cdot)) + v(Tr(H, \cdot))$
\end{center}
\end{lemma}

Similarly, the transshipment problem between $L$ and $S^*\setminus S$ is setup such that each facility in $L$ is closed exactly once and a facility in $S^*\setminus S$ is opened at most once. 
Let $N_s$ be the unused capacity at a facility $s$ in $L$. Since $s$ is a light facility, we must have $N_s + w(Tr(s,\cdot)) \geq U/2$. 
 Let $\theta (s)$ be the cost per unit capacity for making $U/2$ units of capacity available at the node $s$.
 Let $\hat{c}_{st} = w(Sw(s,\cdot))c_{st} + f_t-f_s$ for $t \in S^*\setminus S$, $\hat{c}_{st} = w(Sw(s,\cdot))(c_{st} + \theta (t)) -f_s$ for $t \in L, t\neq s,$ and $\hat{c}_{ss}= \infty$.
 The transshipment problem from $L$ to $( S^*\setminus S)\cup L$  is then defined as follows: 

$$ Minimize \sum\limits_{s \in L, t \in (S^*\setminus S) \cup L} \hat{c}_{st}y_{st}$$

subject to:
\begin{align*}
\sum\limits_{ t \in (S^*\setminus S)\cup L} y_{st} = 1 \quad \forall s\in L \quad (1)\\
 \sum_{s \in L} y_{st} \leq 1, \quad \forall t \in S^* \setminus S  \quad (2)\\
 y_{st} \geq 0, \quad \forall s\in L,  t \in (S^* \setminus S)\cup L \quad (3) \\
\end{align*}

\begin{lemma}[\cite{chudak}]
There is a fractional solution to the transshipment problem with a cost not more than $2v(Sw(L,\cdot)) + v(Tr(L,\cdot)) -c_f(L) + c_f(S^* \setminus S)$.
\end{lemma}
 \begin{proof}
 Using the fact that $N_s + w(Tr(s,\cdot)) \geq U/2$, {\em swap paths} and {\em transfer paths} provide a flow whose cost is no more than $2 v(Sw(L, \cdot)) +  v(Tr (L, \cdot)) + c_f(S^* \setminus S) + c_f(L)$.
 \end{proof}

When $y(s, t) = 1, $ for $t \in S^* \setminus S$, operations are defined in the same manner as in the case of heavy facilities. When $y(s, t) = 1, $ for $t \in L$, we drop $s$ assigning $w(Sw(s, \cdot))$ demands of $s$ to $t$ and the same amount of demands of $t$ to other facilities in $S$ using transfer paths at a total cost at most $\hat{c}_{st}$. Assign the remaining $w(Tr(s, \cdot))$ demands of $s$ to other facilities in $S$ at a cost of $v(Tr(s, \cdot))$.

\begin{lemma}[\cite{chudak}]
\label{lem:unif-light}
If no  swap/delete operation is admissible, then
 \begin{center}
 \vspace{-.3cm}
$c_f(L) \le c_f (S^* \setminus S) + 2 v(Sw(L, \cdot)) + 2 v(Tr(L, \cdot))$
\end{center}
\end{lemma}

Adding the results of lemma~\ref{lem:unif-heavy} and~\ref{lem:unif-light},  we get 
\begin{center}
\vspace{-.3cm}
$c_f(S\setminus S^*)\leq 3c_f(S^*\setminus S) + 2(c_s(S)+c_s(S^*))$ 
\end{center} 
Adding $c_s(S) +c_f(S \cap S^*) $ to both sides and using the bound on the service cost from Lemma~\ref{lem-serv}, we get 
\begin{center}
\vspace{-.5cm}
$c(S) \leq 6c_f(S^*)+ 5c_s(S^*) $
\end{center}

\subsection{Local Search Algorithm for UnifFLPP}
\label{algo-unif}

One may be tempted to add a dummy facility $N$ of zero opening cost and large capacity ($\sum_j d_j$) with $p_j$ as the service cost of serving client $j$ from $N$ and then applying CW on it. The problem with this approach is that the set of penalties do not satisfy triangle inequality. As a result, if $N$ is an internal node on a path in path decomposition, the cost of an $(s, t)$ edge can not be bounded by the cost of the $(s, t)$ paths. 

For CapFLPP the optimal assignment of clients to a given set $S$ of facilities  can be computed by solving a min-cost flow problem with penalty (MCFP) (MCFP can be solved by introducing a dummy supply node with large capacity ($\sum_j d_j$), an edge from client $j$ to the dummy node with per unit cost $p_j$ and solving min cost flow on it). Therefore to solve the CapFLPP also we only need to determine a good subset $S\subseteq F$ of facilities. The cost of the solution $S$ is $c(S) = c_f(S) + c_s(S) + c_p(S)$, where $c_f(S)$ is the facility cost, $c_s(S)$ is the service cost and $c_p(S)$ is the penalty paid by the solution $S$.

The operations remain the same as in CW except that now MCFP is used to assign the clients instead of MCF. 

\textbf{Path Decomposition}: To handle penalties we make the following modifications -  we introduce a dummy facility node, say $N$,  to reflect the penalty paid in the solution $S$. The opening cost of $N$ is $0$ and it has a capacity $u_N = \sum_j d_{j}$. $x(N,j)$ denotes the number of units for which the penalty is paid by client $j$  in $S$.  Thus $x(N, j) = d_j - x(\cdot,j)$.  Now the total flow entering any client exactly equals its demand $d_j$.
The dummy facility in the optimal solution $(S^*, x^*)$ is denoted by $N^*$. Thus, $x^*(N, j) = 0$ and $ x(N^*,j) = 0 \ \forall\  j$.  Note that the presence of $N$ and $N^*$ allows us to maintain flow conservation for each client. 

 By standard path stripping arguments, we decompose the flow into a set of paths $P$ and cycles $C$. Each path starts at a vertex in $S \cup N$ and ends at a vertex in $S^* \cup N^*$. Note that $N$ and $N^*$ cannot be internal vertices on any path as $N$ has no incoming edge and $N^*$ has no outgoing edge. Note that this is the most critical requirement  as penalties do not satisfy triangle inequality. Introducing a single dummy facility ($N$) in the beginning would have allowed $N$ to be in $S \cap S^*$ and hence be an internal vertex on $(s, t)$ paths. Clearly, no cycle passes through a facility outside $S \cap S^*$.  Cycles must have cost $0$ as both $x$ and $x^*$ are minimum cost transshipments and they can be eliminated by augmenting the flows along them. We clearly have,\\
 
 $\>\>$  $\sum_P c(P)w(P) \leq c_s(S)+c_p(S)+c_s(S^*)+c_p(S^*)$.

Let $P(s, t)$ be the set of paths starting at a vertex $s$ and ending at a vertex $t$. Further, let $P(s,\cdot)$  denote the set of paths starting at $s$ and $P(\cdot,t)$  denote the set of paths ending at $t$. Note that $w(P(s, t)), s \neq N, t \neq N^*$ correspond to the amount of demand served in both $S$ and $S^*$,  $w(P(s, N^*)), s \neq N$ correspond to the amount of demand served in $S$  but not in $S^*$,  $w(P(N, t)), t \neq N^*$ correspond to the amount of demand served in $S^*$ but not in $S$ and  $w(P(N, N^*))$ correspond to the amount of demand not served in both $S$ and $S^*$.

Add operation bounds the sum of service cost and the penalty cost as given in the following lemma:
\begin{lemma}
\label{lem-servp}
The sum of service costs $c_s(S)$ and the penalty costs $c_p(S)$ of a locally optimal solution $(S,x)$ is at most the total cost $c(S^*)$ of the optimum solution $(S^*,x^*)$ i.e.

$c_s(S)+c_p(S) \leq c_s(S^*)+c_p(S^*)+c_f(S^*)  =  c(S^*)$.
\end{lemma}
\begin{proof}
Let $C_1$ be the set of demands served in $S$ and $S^*$ both, $C_2$ be the set of demands served in $S$ but not in $S^*$, $C_3$ be the set of demands served in $S^*$ but not in $S$ and $C_4$ be the set of demands not served in any of $S$ and $S^*$. Let $o_j$ and $s_j$ be the service costs of $j$ in $S^*$ and $S$ respectively. Since $add(t)$ operation does not reduce the cost, we have the following:
\vspace{-0.3cm}
$$f_t + \sum_{j \in C_1, served by t in S^*}  (o_j  - s_j)  + \sum_{j \in C_3, served by t in S^*}  (o_j  - p_j)  \ge 0$$
\vspace{-0.4cm}
$$\sum_t f_t + \sum_{j \in C_1}  (o_j  - s_j)  + \sum_{j \in C_3}  (o_j  - p_j)  \ge 0$$
\vspace{-0.3cm}
For $j \in C_2, s_j \le p_j$ for else we would not have served it in $S$. Thus $\sum_{j \in C_2} s_j \le \sum_{j \in C_2} p_j,  \&$
 $$ \sum_{j \in C_1 \cup C_2} s_j  + \sum_{j \in C_3} p_j \le c_f(S^*) + \sum_{j \in C_1 \cup C_3} o_j +  \sum_{j \in C_2} p_j$$

Adding $\sum_{j \in C_4} p_j$ to both the sides we get the desired result.



  \end{proof}

\subsection{Bounding the facility cost for UnifFLPP} 

 {\em Transfer paths} and {\em swap paths} are defined in the same way as in CW.  Let $s \in S \setminus S^*$ be a facility with $w(Sw(s,.)) > 0$.  $s$ is said to be {\em heavy} if $w(Sw(s, \cdot)) \geq  U/2$ otherwise it is called {\em light}.
In our path decomposition, in addition to  {\em transfer paths} and {\em swap paths} we also have another type of paths, the paths that start at a facility $s$ in $S$ and end at $N^*$. We call such paths the {\em penalty paths} (note that we don't need to care about the paths starting at $N$ and we simply ignore them) .  Consider one such penalty path $P$ starting at $s \in S$. Let $s'$ be the facility just before $N^*$ on this path. Then clearly $s' \in S$. Also, clearly $P$ corresponds to some client $j$ of $s'$. Let $Pen_j(s, s')$ be the set of {\em penalty paths} that start at $s$ and correspond to client $j$ of $s'$, $Pen_j(s) = \sum_{s'} Pen_j(s, s')$ be the set of all the {\em penalty paths} that start at $s$ and correspond to client $j$  and $Pen(s)$ denotes the set of all the {\em penalty paths} that start at $s$. Let $v(Pen(H))$ and $v(Pen(L))$ denote the cost of penalty paths for the heavy facilities $H$ and the light facilities $L$ respectively.

The transshipment problem for the heavy facilities is defined exactly in the same manner as in CW. Note that we do not need to care about $N$ and that we can open $N^*$ as many number of times as we want, as the facility opening costs of both $N$ and $N^*$ are zero.  We make the following modification in the assignment of clients: when $y(s, t) = 1$, we perform $swap(s, t)$. Assignments are made as follows:

\begin{itemize}
\item Assign $w(Sw(s, \cdot))$ demands of $s$ to $t$ at cost at most $\hat{c}_{st}$.  
\item  Assign $w(Pen(s))$ demands of $s$ to other facilities in $S$ via the penalty paths. To accommodate the demands of $s$ at these facilities, room is made available by paying penalty for some of the demands assigned to them. This is explained as follows: consider a path $P$ in $Pen_j(s)$ for some $j$ such that $w(Pen_j(s)) > 0 $. Let $s'$ be the facility just before $N^*$ on this path.  Assign $w(Pen_j(s, s'))$ demands of  $s$ to $s'$ and assign the same amount ($w(Pen_j(s, s'))$) of demands of client $j$ of $s'$ to $N^*$ (i.e. penalty is paid by client $j$ served by $s'$ for this much demand). This is done for all $j$ and $s'$. All this can be done at a cost no more than $v(Pen(s))$  . 
\item Assign the remaining $w(Tr(s, \cdot))$ demands of $s$ to other facilities in $S$ via the transfer paths at a cost of $v(Tr(s, \cdot))$. 
\end{itemize}

 Thus, we get the following lemma:


\begin{lemma} 
\label{lem:heavy:penalty}
If there is no admissible swap operation then

$c_f(H) \leq 2 v(Sw(H, .)) + 2 c_f (S^* \setminus S)  + v (Tr(H, .)) +  v(Pen((H)) $

\end{lemma}

In case of the light facilities, it is no longer true that the unused capacity $N_s$ of a facility in $L$ plus the total flow $w(Tr(s, \cdot))$ on its transfer paths must be at least $U/2$. In case when $ N_s + w(Tr(s, \cdot)) < U/2$, we will have to use its penalty paths. Thus $\theta(s)$ is defined as follows:

$\theta(s) = 0$, if $N_s \geq U/2$.

$\>\>\>\> (0.N_s + \frac{v(Tr(s,\cdot))}{w(Tr(s,\cdot))}(U/2-N_s))/(U/2)$, if $N_s < U/2 \le N_s + w(Tr(s,\cdot))$ and,

$\>\> (0.N_s + v(Tr(s,\cdot) + \frac{v(Pen (s))}{ w(Pen(s))} \cdot(U/2 - N_s -  w(Tr(s,\cdot)))/(U/2)$, otherwise. 

In the second case when $N_s + w(Tr(s,\cdot)) \ge U/2$, $\theta (s)(U/2) \leq v(Tr(s,\cdot)) \leq v(Tr(s,\cdot)) +  v(Pen (s))$.
In the third case when $N_s + w(Tr(s,\cdot)) + w(Pen(s)) \ge U/2$, $\theta (s)(U/2) \leq v(Tr(s,\cdot)) +  v(Pen (s))$. Hence in either case, $\theta (s)(U/2) \leq v(Tr(s,\cdot)) +  v(Pen (s))$.\\

The transshipment problem for the light facilities is then set up in the same manner as in CW with the new values for $\theta(s)$.

\begin{lemma}
There is a fractional solution to the transshipment problem with a cost not more than $2v(Sw(L,\cdot)) + v(Tr(L,\cdot)) + v(Pen(L)) -c_f(L) + c_f(S^* \setminus S)$
\end{lemma}
 
 \begin{proof}
For each $t \in S^* \setminus S$, let $s= argmin_{ s' \in L} (c_{s't} + \theta (s')) $.We call $s$ as the primary facility of $t$ and denote it by $ \pi (t) $.$c_{st} + \theta (s)$  denotes the cost of assigning to $s$, unit demand which was supposed to be assigned to $t$ when some other facility $i$ is closed; $\theta(s)$ denotes the cost of shifting that unit demand from $s$ using transfer paths. The fractional solution $\tilde{y}$ is then defined in the same manner as in CW but we will repeat it here for the sake of completeness: \\

$\tilde{y}_{st}= w(Sw(s,t))/w(Sw(s,\cdot))$ if $t \in S^*\setminus S$ and $s=\pi(t)$ .\\
$\>\>\>\>\>=0$,  if $t \in S^*\setminus S$ and $s \ne \pi(t)$.\\

$ \tilde{y}_{si}= \sum\limits_{t \in S^* \setminus S:i=\pi (t), s\neq \pi(t)}w(Sw(s,t))/w(Sw(s,\cdot)), i \in L$.

Clearly, $\sum\limits_{ t \in (S^* \setminus S)\cup L} \tilde{y}_{st} = 1 \forall s\in L$.
  Also, since for a $t \in S^*\setminus S, \pi (t)$ is unique therefore $\tilde{y}_{st} > 0$ for at most one $s \in L$. Thus $\sum_{s \in L}\tilde{y}_{st} \leq 1$. Also, for $s \in l, i= \pi (t), t \in S^* \setminus S$, $\tilde{y}_{si} > 0$  implies that
\[ \tilde{c}_{si} = w(Sw(s,\cdot))(c_{si}+\theta (i) )-f_s \leq w(Sw(s,\cdot))(c_{st}+c_{it}+ \theta (i) )-f_s 
\leq w(Sw(s,\cdot))(2c_{st}+\theta (s))) -f_s
 \]

Thus, the cost of this fractional solution 
\[
 \sum\limits_{s \in L, t \in (S^* \setminus S)\cup L} \hat{c}_{st}\tilde{y}_{st}
\]
\[ 
\leq \sum\limits_{s \in L, t \in S^* \setminus S, s=\pi (t)} [ w(Sw(s,\cdot))c_{st} + f_t- f_s] w(Sw(s,t))/w(Sw(s,\cdot))  
 \]
 \[
+ \sum\limits_{s \in L, t \in S^* \setminus S, s \neq \pi (t)}[w(Sw(s,\cdot))(2c_{st}+\theta (s)) - f_s]w(Sw(s,t)/w(Sw(s,\cdot)) 
\]
\[
\leq \sum\limits_{s \in L, t \in S^* \setminus S} w(Sw(s,\cdot))(2c_{st}+\theta (s)) ]w(Sw(s,t)/w(Sw(s,\cdot)) + \sum\limits_{t \in S^* \setminus S} f_t - \sum\limits_{s \in L }f_s
\]
\[
\leq \sum\limits_{s \in L, t \in S^* \setminus S} 2c_{st}w(Sw(s,t))  + \sum\limits_{s \in L} (U/2) \theta (s)) + \sum\limits_{t \in S^* \setminus S} f_t - \sum\limits_{s \in L }f_s
\]
\[
\leq \sum\limits_{s \in L, t \in S^* \setminus S} 2c_{st}w(Sw(s,t)) + \sum_{s \in L} v(Tr(s,\cdot)) + \sum_{s \in L} v(Pen(s)) - c_f(L) + c_f(S^* \setminus S)
\]
\[
\leq 2v(Sw(L,\cdot)) + v(Tr(L,\cdot))  + v(Pen(L)) -c_f(L) + c_f(S^* \setminus S) 
\]
 \end{proof}


When $y(s, t) = 1, $ for $t \in S^* \setminus S$, operations are defined in the same manner as in the case of heavy facilities. When $y(s, t) = 1, $ for $t \in L$, we drop $s$. The assignments are made as follows:

\begin{itemize}
\item Assign $w(Sw(s, \cdot))$ demands of $s$ to $t$. If $t$ has sufficient space to accommodate this much demand of $s$ (i.e. $N_t \le w(Sw(s, \cdot))$) we are done else let $rem(s)$ denote the amount of demand that could not be accommodated in $t$ i.e $rem(s) = N_t - w(Sw(s, \cdot))$, we need to make room to accommodate this much demand at $t$. This is done by assigning $\min \{rem(s), w(Tr(t, \cdot)) \}$ demands of $t$ to other facilities in $S$ using transfer paths. Update $rem(s)$ as $rem(s) = rem(s) - \min \{rem(s), w(Tr(t, \cdot)) \}$. If $rem(s) > 0$, we use the penalty paths to vacate space at $t$. For the $rem(s)$ demands assigned to $t$ we do the following: consider a path $P$ in $Pen_j(t)$ for some $j$ for which $Pen_j(t) > 0$. Let $t'$ be the facility just before $N^*$ on this path.  Assign $\min \{rem(s), w(Pen_j(t, t'))\}$ demands of  $t$ to $t'$ and assign the same amount ( $\min \{rem(s), w(Pen_j(t, t'))\}$) of demands of client $j$ of $t'$ to $N^*$ (i.e. pay penalty for them). Update $rem(s)$ as $rem(s) = rem(s) - \min \{rem(s), w(Pen_j(t, t'))\}$. This is repeated for all $j$ and $t'$ until we have made sufficient space in $t$.  All this can be done at a total cost at most $\hat{c}_{st}$.
\item  Assign $w(Pen(s))$ demands of $s$ to other facilities in $S$ via the penalty paths. To accommodate the demands of $s$ at these facilities, room is made available by paying penalty for some of the demands assigned to them. This is done in the same manner as explained for heavy facilities at a cost no more than $v(Pen(s))$  . 
\item Assign the remaining $w(Tr(s, \cdot))$ demands of $s$ to other facilities in $S$ via the transfer paths at a cost of $v(Tr(s, \cdot))$. 


\end{itemize}

\begin{lemma} If there is no admissible swap/delete operation then
\label{lem:light:penalty}

$c_f(L) \leq 2 v(Sw(L, .)) +  c_f (S^* \setminus S)  + 2 v(Tr(L, .)) +   2 v(Pen(L))$
\end{lemma}

Adding the results of Lemmas~\ref{lem:heavy:penalty} and~\ref{lem:light:penalty} we get : 
 \begin{center}

$c_f(S\setminus S^*)\leq 3c_f(S^*\setminus S) + 2(c_s(S)+c_p(S) +c_s(S^*)+ c_p(S^*))$ 
\end{center} 

Adding $c_s(S) +c_f(S \cap S^*) $ to both sides and using the bound on the sum of service cost and penalty cost (Lemma~\ref{lem-servp}), we get 
\begin{center}

$c(S) \leq 6c_f(S^*)+ 5c_s(S^*)+5 c_p(S^*)$
\end{center}

\section{Non-uniform CapFLPP}
In this section, we consider the facility location problem with penalties when the capacities are not necessarily uniform. We first briefly review the local search algorithm of Pal-Tardo-Wexler\cite{paltree} and then show how it can be extended to integrate penalties. We add a new operation $delete(s)$. 

\subsection{Preliminaries and Previous Work}
As is true for the uniform capacity case, for a given subset $S\subseteq F$, the optimal assignment of clients to the set of facilities in $S$ can be computed by solving a mincost flow problem. 
Pal-Tardos-Wexler~l\cite{paltree} suggested a local search algorithm to find a good approximate solution for non-uniform case  (without penalties). Starting with a feasible solution $S$ the following operations are performed to improve the solution if possible.

\begin{itemize}

\item \textbf{add(t):} $S \gets S \cup \{t\}, t \in F$.

\item \textbf{open(t,T):}  $S\gets S\cup \{t\} \setminus T, t \in F, T\subseteq S  \setminus \{t\}$. In this operation a facility $t \in F$ is added and a subset of facilities $T\subseteq S  \setminus \{t\}$ is closed. All the demands served by the facilities in $T$ are simply reassigned to $t$.

Any demand served by $s' \in S \setminus T$ is not affected by this operation. The operation can be performed in polynomial time by solving a knapsack problem approximately.

\item \textbf{close(s,T):} $S\gets S\cup T \setminus \{s\}, s\in S, T\subseteq F\setminus \{s\}$. In this operation a facility $s\in S$ is closed and a subset of facilities $T$ is opened. An estimated cost (which is an upper bound to the actual cost) for the operation is computed, in which it is assumed that a client which was assigned to $s$ in the solution $S$ will now be assigned to some $t\in T$ and that this reassignment costs  at most $c_{st}$. 

Any demand served by $s' \in S \setminus \{s\}$ is not affected by this operation.
The operation can be performed in polynomial time by solving a covering knapsack problem approximately.
\end{itemize}

When a facility already opened, is opened, we only use its free capacity. Add operation provides  a bound on the service cost as in Lemma~\ref{lem-serv}.

To bound the facility costs, they use path decomposition, swap graphs  and the integrality of the transshipment polyhedra to define the valid operations. Following transshipment problem is set up: we require the flow out of every facility $s \in S\setminus S^*$ to be exactly $x(s,\cdot)$ and a flow entering a facility $t \in S^*$ to be at most the free capacity of $t$, $u_t - x(t,\cdot)$. 

minimize  $\sum\limits_{s \in S \setminus S^*, t \in S^*} c_{st} y(s,t)$\\
  
  subject to
  
$\sum\limits_{t \in S^*} y(s,t) = x(s,\cdot)$  $\forall s \in S \setminus S^*$ \\

$\sum\limits_{s \in S \setminus S^*} y(s,t)  \leq  u_t-x(t,\cdot)$  $ \forall t\in S^*$\\

$y(s,t) \geq 0 $  $\forall s \in S \setminus S^*, t \in S^*$  \\

\begin{lemma}[\cite{paltree}]
There exists a feasible flow whose cost is no more than $c_s(S) + c_s(S^*)$.
\end{lemma}
\begin{proof}
Path decomposition provides a flow whose cost is no more than $c_s(S) + c_s(S^*)$.
\end{proof}

Consider the minimum cost flow of the above transshipment problem. Remove cycles, if any (of zero cost) by augmenting flow along it. 
Consider the forest  formed by the edges with nonzero flow in the flow $y$ so obtained. Consider each tree $T$ rooted at some vertex in $S^*$. For a facility $t$,  let $C(t)$ denote the children of $t$. For a non-leaf vertex $t \in S^*$,  let $T_t$ be the subtree, rooted at $t$, containing all children and grandchildren of $t$. Thus height of $T_t$ is at most $2$. A facility $s \in C(t)$ is said to be {\em heavy} if $y(s,t) > y(\cdot,t)/2$ and it is \textit{light} otherwise. A light facility $s$ is \textit{dominant} if $y(s,t) \geq y(s,\cdot)/2$ and otherwise it is said to be \textit{non-dominant}. The set $C(t)$ is partitioned into three sets:  the set $H$ of heavy facilities , the set $LND$ of light non-dominant facilities  and the set $LD$ of light dominant facilities.  Following set of operations are defined for each subtree $T_t$:

\begin{itemize}
\item When all children of $t$ are leaves  ( i.e. $T_t$ is of height $1$),  consider $open( t, C(t))$. Otherwise,
\item For $s \in H$, consider $close(s, C(s) \cup\{t\})$.
\item For facilities in $LD$,
          \begin{itemize}
          \item  if $H \neq \phi$, consider $open(t, LD)$ (since $H \neq \phi$, all the demands of all the facilities in LD can be assigned to $t$ without exceeding its capacity) else
          \item $LD$ is divided into at most three sets $LD_1, LD_2$, and $LD_3$ such that the operations $open(t, LD_i)$ are feasible 
for $ i = 1, 2$ and $3$.
          \end{itemize}
 \item Facilities $s$ in $LND$ are arranged in the increasing order of their $y(s, t)$ values. i.e. $y(s_1, t) \le y(s_2, t) \le \ldots y(s_k, t)$. For $ 1 \le i < k$, consider $close(s_i, C(s_i) \cup C(s_{i+1})$;  for $s_k$, consider $close(s_k, C(s_k) \cup t)$.
 \end{itemize}

The above operations ensure that a facility in $S^*$ is opened at most $6$ times and flow edges are used at most twice. Hence we get the following:

$c_f(S\setminus S^*) \leq 6c_f(S^* \setminus S) + 2\alpha - c(S\setminus S^*)$, 
where $\alpha$ is the cost of the optimal transshipment. Adding $c_s(S) +c_f(S \cap S^*) $ to both sides, using the fact that $\alpha \leq c_s(S) + c_s(S^*)$ and  the bound on the service cost, we get 
\begin{center}
\vspace{-.3cm}
$c(S) \leq 9c_f(S^*)+ 5c_s(S^*) $ 
\end{center}
\subsection{Local Search Algorithm for CapFLPP}

To incorporate penalties we modify the operations of PTW as follows:
\begin{itemize}

\item \textbf{add(t):} $S \gets S \cup \{t\}, t \in F$. MCFP is used instead of MCF to assign the clients optimally

\item \textbf{delete(s):} $S \gets S \setminus \{s\}, s \in S$. MCFP is used to assign the clients optimally.

\item \textbf{open(t,T):}  This operation is same as that in PTW.
\item \textbf{close(s,T):} Close one facility $s \in S$ and open a set of facilities $T \subseteq F \setminus \{s\}$ i.e $S \gets S \setminus \{s\} \cup T$. Here we reassign some of the demands served by $s$ to a facility in $T$, some demands are assigned to some other facilities $s' \in S$ and space is made available at these facilities by paying penalty for these many demands served by $s'$ and some demands of $s$ pay penalty ( this actually corresponds to the case when $s' =s$). 

\end{itemize}
 
 \begin{lemma}[\cite{paltree}]
 Open(t, T)  operation can be performed in polynomial time when the demands and capacities are small integers.
 \end{lemma}

For \textbf{$close(s,T)$}: We first use a rough upper bound on the cost of rerouting demands from $s$ to $T$. For a client $j$ the change in service cost by rerouting it from $s$ to $t \in T$ is $c_{tj}-c_{sj}$ which we estimate by $c_{st}$, which is an upper bound on the change by using the triangle inequality (this is same as the estimated cost used in PTW). Simialrly, the change in service cost by rerouting a unit demand from $s$ to $s' \in S$ is estimated by $c_{s s'}$.  Further, a unit demand of client $j$ served by $s'$ pays a penalty of $p_j$ to make space for a unit demand of $s$. When some client $j$ pays penalty to make space for a unit demand of $s$, we say that we pay penalty for the unit demand of $s$ ( note that this penalty is being paid indrectly by the clients of $s'$ which could be $s$ itself at times). Let this amount be denoted by $r$. Since we don't know $r$, we make a guess for it. Then we say that the operation is feasible if the total free capacity of the facilities in $T$ is at least $x(s, .) - r$.
We can then view the problem as a special case of facility location problem in which there is only one client to serve with total demand $x(s, .) - r$.

 If we decide to pay the penalty for a client $j$ then the change in cost for a demand of that client is $p_j-c_{sj}$. We notice that this change is independent of the set $T$. We know that $p_j - c_{sj}$ are all positive numbers for if $p_{j} < c_{sj}$ our solution wouldn't be serving these clients in the first place. 

\begin{lemma}[\cite{paltree}]
The special case of facility location problem with a single client $j$ with total demand $d_j$ can be solved in time polynomial in $d_j$ and the number of facilities $n$.
\end{lemma}



Since a unit demand of every client $j$ of every facility $s'$ is a candidate to pay penalty for one unit of demand of $s$ when $s$ is closed, we compute the quantity $c_{s s'} + p_j$ for all client $j$ of every facility $s' \in S$. We order the demands in the increasing order of  this quantity and select the demands that add up to  $r$, in that order, to pay the penalty on behalf of $r$ units of demands served by $s$. Suppose $r_{s' j}$ demands of client $j$ served by $s'$ are selected, then re-assign $r_{s' j}$ demands of $s$ to $s'$ and pay penalty for the same amount of demand of $j$.

Since there are $d_j$ choices for $r$, we repeat the above process $d_j$ number of times and then choose the solution with the minimum cost. Thus {\em close} operation can be performed in time polynomial in $d_j$ and the number of facilities $n$. Can we get an approximate solution in strongly polynoimal time?

\subsection{Bounding the costs for CapFLPP}

As before, add operation bounds the sum of service cost and the penalty cost as given in Lemma 6.
To bound the facility costs, we use the path decomposition as described in section~\ref{algo-unif}. However, to define the valid operations we setup the following transshipment problem.
 We make $nm$ copies of $N^*$, one for each client and a facility. Let $N^*_{s,j}$ denote a copy of $N^*$ corresponding to client $j$ and facility $s \in S \setminus S^*$.  Let $y(s,t)$ denote the flow from $s$ to $t \in S^*$ and $y(s, N^*_{s',j})$ denote the flow going from $s$ to $N^*_{s',j}$ through client $j$.
 We now formulate the problem as follows:\\
 
  minimize  $\sum\limits_{s \in S \setminus S^*, t \in S^*} c_{st} y(s,t) + \sum\limits_{s, s' \in S \setminus S^* , j} (c_{s s'} + p_j) y(s,N^*_{s',j})$\\
  
  subject to
  
$\sum\limits_{t \in S^*} y(s,t) + \sum_{s' \in S \setminus S', j} y(s,N^*_{s', j})= x(s,\cdot)$  $\forall s \in S \setminus S^*$ \\

$\sum\limits_{s \in S \setminus S^*} y(s,t)  \leq  u_t-x(t,\cdot)$  $ \forall t\in S^*$\\

$y(s,t) \geq 0 $  $\forall s \in S \setminus S^*, t \in S^*$  \\

$ 0 \le y(s,N^*_{s', j}) \le w(Pen_j(s, s'))$  $\forall j \in C, s, s' \in S \setminus S^*$ \\


Following lemma shows that there is a low cost feasible flow.

\begin{lemma}
There is a flow with cost at most $c_s(S)+c_p(S)+c_s(S^*)+c_p(S^*)$.
\end{lemma}

\begin{proof}: We use the path decomposition described in section~\ref{algo-unif}  to construct such a flow. For a pair of facilities $s \in S \setminus S^*$ and $t \in S^* $, set $y(s,t)=  w(P(s,t)) = \sum_{p \in P(s,t)} w(p)$ and $y (s,N^*_{s',j}) = w(Pen_j(s, s')) = \sum_{p \in Pen_j(s, s')} w(p)$. The total flow leaving every vertex $s \in S \setminus S^*$ is then $x(s,\cdot)$. Also, total flow entering any vertex $t \in S^*$ is at most $max\{ x^*(t,\cdot)-x(t,\cdot), 0 \}$ which is at most $u_{t}-x(t,\cdot)$. Thus,  $y$ defines a feasible flow. 

 The cost of flow on an $(s,t)  (t \in S^*)$ edge is $c_{st} y(s,t)$. By triangle inequality $c_{st} \leq c(p)$ for every path $p \in P(s,t)$. Next, consider the edge $(s,N^*_{s', j})$. The cost of flow on this edge is $ (c_{s s'} + p_j) y(s, N^*_{s', j})$. Again by triangle inequality, for every path  $p \in Pen_j(s, s'), c_{s s'} + p_j \leq c(p)$ (note that triangle inequality is applied on the path from $s$ to $s'$ and $p_j$ does not particiapte in it). 
 The cost of the solution is then bounded by\\
 $ \sum\limits_{s \in S \setminus S^*\, t \in S^*} c_{st}y(s,t)  + \sum\limits_{s, s' \in S \setminus S^*,j} (c_{s s'} + p_j) y(s, N^*_{s',j})$
   $\leq    \sum_{p \in P} c(p)w(p)$
 $\leq c_s(S)+c_p(S)+c_s(S^*)+c_p(S^*)$ 
\end{proof}


Consider the graph of edges with a positive flow in the optimal solution of the above transshipment problem. If the graph contains cycles involving $N^*_{s, j}$, break them by making further copies of $N_{s, j}$ (total of at most $n$). To avoid complicating the notation further we will simply call them copies of $N_{s, j}$, rest will be clear from the context. Other cycles can be eliminated by augmenting flow along them as the cycles are of cost zero.

 
  
  






\textbf{Bounding Facility Costs}:  To bound the facility costs, we will use the same operations as used by Pal-Tardos-Wexler. 
 Consider the forest  formed by the edges with nonzero flow in the flow $y$ obtained after removing cycles. Consider each tree $T$ rooted at some vertex in $S^*$. For a facility $t$,  let $C(t)$ denote the children of $t$ and $\hat{C}(t)$ denote the set of children of $t$ minus the copies of $N^*$. For a non-leaf vertex $t \in S^*$,  let $T_t$ be the subtree, rooted at $t$, containing all children and grandchildren of $t$. 
For a facility $t \in S^*$, {\em heavy}  and {\em light} facilities are defined as in PTW.  A facility $s \in S \setminus S^*$ is \textit{dominant} if $y(s,t) \geq \frac{1}{2} y(s,\cdot)$ and otherwise it is said to be \textit{non-dominant}. A facility is {\em light dominant} if it is {\em light} and {\em dominant}. Similarly, A facility is {\em light non-dominant} if it is {\em light} and {\em non-dominant}.

In what follows, when we say assign $y(s, N^*_{s', j})$ demand of $s$ to $N^*_{s', j}$ we mean the following: consider a path $P$ in $Pen_j(s, s')$.  Assign $y(s, N^*_{s', j})$ demands of  $s$ to $s'$ and assign the same amount ($y(s, N^*_{s', j})$) of demands of client $j$ served by $s'$ to $N^*$ (i.e. penalty is paid by client $j$ served by $s'$ for this much demand).  Note that this is feasible since  $y(s, N^*_{s', j}) \le w(Pen_j(s, s'))$.

When $t \ne N^*_{s,j}$ for any $s$ and $j$, following set of operations are used:
\begin{itemize}
\item When all children of $t$ are leaves  ( i.e. $T_t$ is of height $1$), consider $open( t, C(t))$,  otherwise,
\item For $s \in H$, consider $close(s, \hat{C}(s)  \cup\{t\})$.  Assign $y(s, o)$ demand served by $s$ to $o$ for $o \in \hat{C}(s) \cup \{t\}$ and $y(s, N^*_{s', j})$ demand to $N^*_{s', j}$,
for all $j$ and for all $s'$.

\item For facilities in $LD$, 
          \begin{itemize}
          \item  if $H \neq \phi$, consider $open(t, LD)$. Assign $\sum_{o \in C(s) \cup t} y(s, o)$ demand served by $s$ (in $LD$) to $t$, else
          \item $LD$ is divided into at most three sets $LD_1, LD_2$, and $LD_3$ such that the operations $open(t, LD_i)$ are feasible 
for $ i = 1, 2$ and $3$. Assign $\sum_{o \in C(s) \cup t} y(s, o)$ demand served by $s$ (in $LD_i$) to $t$.
          \end{itemize} 
 \item Facilities in $LND$ are arranged in the increasing order of their $y(s, t)$ values. i.e. $y(s_1, t) \le y(s_2, t) \le \ldots y(s_k, t)$. For $ 1 \le i < k$, consider $close(s_i, \hat{C}(s_i) \cup \hat{C}(s_{i+1}) ))$ and for $s_k$, consider $close(s_k, \hat{C}(s_k) \cup t )))$;   Assign $y(s_i, o)$ demand served by $s$ to $o$ for $o \in \hat{C}(s_i)$, assign $y(s_i, N^*_{s', j})$ demand to $N^*_{s', j}$ for all $s'$ and $j$.  For $i = k$, $y(s_i, t)$ demands are reassigned  to $t$ itself but assignment is a little tricky for $i < k$. For $i < k$, some of  $y(s_i, t)$ demands are reassigned to the facilities in $\hat{C}(s_{i+1}))$ and the rest are assigned to $s_{i+1}$. Space is vacated at $s_{i+1}$ by assigning $y(s_{i+1}, N^*_{s', j})$ demand to $N^*_{s', j}$ for all $s'$ and $j$.
 

 \end{itemize}
 
When $t = N^*_{s', j}$ for some $s'$ and $j$, following operations are used:
\begin{itemize}
\item When all children of $t$ are leaves  ( i.e. $T_t$ is of height $1$) , consider $delete(s)$ for all $s \in C(t)$ and assign $y(s, N^*_{s', j})$ demand to $N^*_{s', j}$, else


\item  consider $close(s, \hat{C}(s))$, assign $y(s, o)$ demand of $s$ to $o$ for $o \in \hat{C}(s)$ and $y(s, N^*_{\hat{s},j})$ demand to $N^*_{\hat{s},j}$,
 for all $j$ and all $\hat{s}$.
 \end{itemize}


In the above operations a facility is opened at most thrice as a parent for the facilities in $H \cup LD$,  at most once as a parent of a facility in LND, and at most twice as a child of a LND facility. Note that a facility that is a grand child of a copy of $N^*$ opens at most once as a child. Thus, a facility in $S^*$ is opened at most $6$ times. $(s, t)$ and $(s, N^*_{s, j})$ edges are used at most twice in a similar manner as in PTW.  
Hence we get the following:

$c_f(S \setminus S^*) \leq 6c_f(S^* \setminus S) + 2\alpha' - c(S\setminus S^*)$, where $\alpha'$ is the cost of the optimal transshipment. Adding $c_s(S) + c_p(S) + c_f(S \cap S^*) $ to both sides, using the fact that $\alpha' \leq c_s(S) + c_p(S) + c_s(S^*) + c_p(S^*)$ and  the bound on the sum of service cost and penalty costs, we get 

\begin{center}
$c(S) \leq 9c_f(S^*)+ 5c_s(S^*) + 5 c_p(S^*)$ 
\end{center}

In order to guarantee polynomial running time, we insist on operations that improve the cost by at least $c(S)/p(n,\epsilon)$. To drop the assumption of demands, capacities and profits to be "small" integers, we solve the knapsack problems approximately and we rather insist on operations that improve the cost by at least $c(S)/ (2 p(n,\epsilon))$. 
 We get the weaker bounds of $(6 + \epsilon)$ for UnifFLPP which improves to $(5.83 + \epsilon)$ after scaling and   $(9 + \epsilon)$ for CapFLPP which improves to $(8.532 + \epsilon)$ after scaling. Hence the Theorems~\ref{theo-unif} and~\ref{theo-nonunif}.




\begin{thebibliography}{50}
\bibitem{Korupolu}Korupolu, Madhukar R., C. Greg Plaxton, and Rajmohan Rajaraman. "Analysis of a local search heuristic for facility location problems." \textit{Journal of algorithms} 37.1 (2000): 146-188.
\bibitem{jain2003greedy}Jain, K., Mahdian, M., Markakis, E., Saberi, A., and Vazirani, V. V. (2003). Greedy facility location algorithms analyzed using dual fitting with factor-revealing LP. \textit{Journal of the ACM (JACM)}, 50(6), 795-824.
\bibitem{chudak}Chudak, Fabián A., and David P. Williamson. "Improved approximation algorithms for capacitated facility location problems."\textit{ Integer programming and combinatorial optimization}. Springer Berlin Heidelberg, 1999. 99-113.
\bibitem{paltree}Pal, Martin, Éva Tardos, and Tom Wexler. "Facility location with nonuniform hard capacities." \textit{Foundations of Computer Science, 2001. Proceedings. 42nd IEEE Symposium on}. IEEE, 2001.
\bibitem{charikar2001algorithms}Charikar, M., Khuller, S., Mount, D. M., and Narasimhan, G. (2001, January). Algorithms for facility location problems with outliers. \textit{In Proceedings of the twelfth annual ACM-SIAM symposium on Discrete algorithms} (pp. 642-651). Society for Industrial and Applied Mathematics.
\bibitem{xu2009improved}Xu, Guang, and Jinhui Xu. "An improved approximation algorithm for uncapacitated facility location problem with penalties." \textit{Journal of combinatorial optimization} 17.4 (2009): 424-436.
\bibitem{goemans1995general}Goemans, Michel X., and David P. Williamson. "A general approximation technique for constrained forest problems."\textit{ SIAM Journal on Computing }24.2 (1995): 296-317.
\bibitem{bienstock1993note}Bienstock, Daniel, Michel X. Goemans, David Simchi-Levi, and David Williamson. "A note on the prize collecting traveling salesman problem." \textit{Mathematical programming} 59, no. 1-3 (1993): 413-420.
\bibitem{konemann2006unified} Könemann, Jochen, Ojas Parekh, and Danny Segev. "A unified approach to approximating partial covering problems." \textit{Algorithms–ESA} 2006. Springer Berlin Heidelberg, 2006. 468-479.
\bibitem{hajiaghayi2010prize} Hajiaghayi, MohammadTaghi, Rohit Khandekar, Guy Kortsarz, and Zeev Nutov. "Prize-collecting steiner network problems." I\textit{n Integer Programming and Combinatorial Optimization,} pp. 71-84. Springer Berlin Heidelberg, 2010.
\bibitem{ipco}Aggarwal, Ankit, L. Anand, Manisha Bansal, Naveen Garg, Neelima Gupta, Shubham Gupta, and Surabhi Jain. "A 3-approximation for facility location with uniform capacities." I\textit{n Integer Programming and Combinatorial Optimization}, pp. 149-162. Springer Berlin Heidelberg, 2010. \
\bibitem{arya} Arya, Vijay, Naveen Garg, Rohit Khandekar, Adam Meyerson, Kamesh Munagala, and Vinayaka Pandit. "Local search heuristics for k-median and facility location problems." \textit{SIAM Journal on Computing} 33, no. 3 (2004): 544-562..
\bibitem{bansal20125}Bansal, Manisha, Naveen Garg, and Neelima Gupta. "A 5-approximation for capacitated facility location." \textit{Algorithms–ESA 2012}. Springer Berlin Heidelberg, 2012. 133-144.
\bibitem{AnarXiv}An Hyung-Chan, Mohit Singh, Ola Svensson "LP-Based Algorithms for Capacitated Facility Location" arXiv:1407.3263 [cs.DS]
\end{thebibliography}
\end{document}